\documentclass[oneside, 12pt, a4paper]{article}
\usepackage[utf8]{inputenc}
\usepackage{lipsum}
\usepackage{amsmath}
\usepackage{amssymb}
\usepackage{amsthm}
\usepackage{mathtools}
\usepackage{fancyhdr}
\usepackage{graphicx}
\usepackage{verbatim}
\usepackage{latexsym}
\usepackage{mathrsfs}
\usepackage{tikz}
\usepackage{tikz-cd}
\usepackage{float}
\usepackage{enumerate}
\usepackage{pdfpages}
\usepackage{xcolor}
\usepackage{sectsty}
\usepackage[calcwidth]{titlesec}
\usepackage{stmaryrd}
\usepackage{wallpaper}
\usepackage{setspace}
\usepackage[top = 2 cm, bottom = 3 cm, left = 2 cm, right = 2 cm]{geometry}
\usepackage{bbm}
\usepackage{datetime2}
\usepackage{anyfontsize}

\usepackage[numbers,sort&compress]{natbib}
\usepackage[hidelinks]{hyperref}
\usepackage[noabbrev]{cleveref}

\newcommand{\R}{\mathbb{R}}

\newtheoremstyle{theorems}
  {3pt}
  {3pt}
  {\itshape}
  {}
  {\bfseries}
  {.}
  { }
  {}

\newtheoremstyle{proofparts}
  {3pt}
  {0pt}
  {}
  {\parindent}
  {\scshape}
  {:}
  {\newline}
  {}

\newtheoremstyle{claims}
  {2pt}
  {2pt}
  {}
  {\parindent}
  {\bfseries}
  {.}
  { }
  {}

\theoremstyle{theorems}
\newtheorem{thm}{Theorem}

\newtheorem{cor}[thm]{Corollary}
\newtheorem{prop}[thm]{Proposition}

\theoremstyle{definition}
\newtheorem{defn}[thm]{Definition}

\theoremstyle{proofparts}

\makeatletter
\@addtoreset{proofpart}{thm}
\makeatother

\theoremstyle{claims}

\newtheorem*{claim*}{Claim}

\newcommand{\norm}[1]{\left\Vert #1 \right\Vert}
\newcommand{\abs}[1]{\left\vert #1 \right\vert}

\DeclareMathOperator{\sgn}{sgn}

\definecolor{emphcolor}{rgb}{0,0,1}           

\newcommand{\longip}[3]{\left\langle #1 \middle\vert #2 \middle\vert #3 \right\rangle}

\newcommand{\ud}{\,\textnormal{d}}

\newcommand{\bra}[1]{\left\langle #1 \right\vert}
\newcommand{\ket}[1]{\left\vert #1 \right\rangle}

\let\epsilon\varepsilon
\let\varepsilon\epsilon
\let\eps\epsilon

\title{The BCS Energy Gap at Low Density}
\author{Asbj\o rn B\ae kgaard Lauritsen\thanks{\url{asbjorn.lauritsen@gmail.com}}
\\ Department of Mathematics, University of Copenhagen, \\ 
Universitetsparken 5, 2100 Copenhagen, Denmark}

\begin{document}
\maketitle
\begin{abstract}
\noindent
We show that the energy gap for the BCS gap equation is 
\[
	\Xi = \mu \left( 8 e^{-2} + o(1)\right) \exp\left( \frac{\pi}{2\sqrt{\mu} a}\right)
\]
in the low density limit $\mu \to 0$. 
Together with the similar result for the critical temperature \cite{hainzl.seiringer.scat.length}
this shows that, in the low density limit, 
the ratio of the energy gap and critical temperature is a universal constant 
independent of the interaction potential $V$.
The results hold for a class of potentials with negative scattering length $a$ and no bound states.
\end{abstract}

\section{Introduction and Main Result}
The BCS gap equation at zero temperature
\[
	\Delta(p) = - \frac{1}{(2\pi)^{3/2}} \int_{\R^3} \hat V(p-q) \frac{\Delta(q)}{E(q)} \ud q,
\]
where $E(p) = \sqrt{ (p^2 - \mu)^2 + |\Delta(p)|^2}$, 
is an important part of the BCS theory of superfluidity and -conductivity \cite{bcs.original}.
The function $\Delta$ has the interpretation of the order parameter 
describing pairs of Fermions (Cooper pairs). 
The potential $V$ models an effective local interaction. 
(In the case of superconductivity it is between electrons.)
We will assume that $V \in L^1(\R^3)$, in which case 
$V$ has a Fourier transform given by
$\hat V(p) = (2\pi)^{-3/2} \int_{\R^3} V(x) e^{-ipx} \ud x$.
The chemical potential $\mu > 0$ controls the particle density. 
We study here the limit of low density, i.e. $\mu \to 0$.
In the low density limit, it is known that superfluid/-conducting behaviour
is well described by BCS theory \cite{leggett.diatomic,nozieres.schmitt-rink}.

The low density limit for BCS theory has previously been studied in \cite{hainzl.seiringer.scat.length},
where the critical temperate has been calculated. The critical temperature satisfies that,
for temperatures below the critical temperature the system is in a superfluid/-conducting state.
For temperatures above, it is not.
We will here study the energy gap (at zero temperature)
\[
	\Xi = \inf E(p) = \inf \sqrt{(p^2 - \mu)^2 + |\Delta(p)|^2}.
\]
in this limit. The function $E$ has the interpretation as 
the dispersion relation for the corresponding BCS Hamiltonian, 
and so $\Xi$ has the interpretation of an energy gap, see \cite[Appendix A]{hainzl.hamza.seiringer.solovej}.

The potential non-uniqueness of such functions $\Delta$ gives rise 
to difficulty in evaluating this $\Xi$. However, under the assumption we impose on $V$
(namely that its Fourier transform is non-positive),
it is proved in \cite{hainzl.seiringer.08} that $\Delta$ is unique 
(up to a constant global phase).

An analysis of the energy gap in the low coupling limit is given in \cite{hainzl.seiringer.08}.
There one considers a potential $\lambda V$, for $V$ fixed and $\lambda \to 0$.
In this limit it is shown that the energy gap satisfies $\Xi \sim A \exp(-B/\lambda)$ 
for explicit constants $A$ and $B$. 
We are here instead interested in the limit $\mu \to 0$ for $V$ fixed. 
Similarly as for the critical temperature in the low density limit \cite{hainzl.seiringer.scat.length}
this turns out to be related to the scattering length of $2V$, which we now define.

\begin{defn}[{\cite[Definition 2]{hainzl.seiringer.scat.length}}]
Let $V \in L^{1}(\R^3)\cap L^{3/2}(\R^3)$ be real-valued. By $V(x)^{1/2}$ we will mean
$V(x)^{1/2} = \sgn(V(x)) |V(x)|^{1/2}$. 
Suppose that $-1$ is not in the spectrum of the associated Birman-Schwinger operator 
$V^{1/2} \frac{1}{p^2} |V|^{1/2}$. Then the scattering length $a$ of $2V$ is
\[
	a = \frac{1}{4\pi} \longip{|V|^{1/2}}{\frac{1}{1 + V^{1/2}\frac{1}{p^2}|V|^{1/2}}}{V^{1/2}}.
\]
Here, operators that are functions of $p$ are to be interpreted as multiplication operators in Fourier space.
\end{defn}

\noindent
In \cite[Appendix A]{hainzl.seiringer.scat.length} it is explained, 
why it is sensible to call this a scattering length.

With this, we may now state our main theorem. 
\begin{thm}\label{thm.energy.gap.low.density}
Let $V$ be radial and assume that $V(x)(1 + |x|) \in L^1(\R^3) \cap L^{3/2}(\R^3)$, 
$\hat V \leq 0$, $\hat V(0) < 0$, 
that $\norm{V}_{L^{3/2}} < S_3$, 
and that the scattering length $a < 0$ is negative. 
Then, 
\[
  \lim_{\mu \to 0} \left( \log \frac{\mu}{\Xi} + \frac{\pi}{2\sqrt\mu a}\right) = 2 - \log 8.
\]
\end{thm}

\noindent
That is, in the limit of low density, the energy gap satisfies
\[
  \Xi = \mu \left(8 e^{-2} + o(1)\right) \exp \left(\frac{\pi}{2\sqrt{\mu} a}\right).
\]
This is known in the physics literature \cite{leggett.diatomic}. 
Here $S_3 = \frac{3}{4}2^{2/3} \pi^{4/3} \approx 5.4779$ is the best constant in Sobolev's inequality \cite[Theorem 8.3]{analysis}.
The assumption that $\norm{V}_{L^{3/2}} < S_3$ gives that
$p^2 + \lambda V > 0$ for any $0 < \lambda \leq 1$ by Sobolev's inequality, see \cite[section 11.3]{analysis}.
Thus, by the Birman-Schwinger principle, the operator 
$\lambda V^{1/2} \frac{1}{p^2} |V|^{1/2}$ does not have $-1$ as an eigenvalue. 
Varying $\lambda$ we thus get that the spectrum of $V^{1/2} \frac{1}{p^2} |V|^{1/2}$ 
is contained in $(-1, \infty)$. 
In particular, the scattering length is finite. 
Also, for a $V$ satisfying the assumption it also satisfies the assumptions of \cite[Theorem 1]{hainzl.seiringer.scat.length}.
This states that the critical temperature satisfies
\[
	T_c = \mu \left( \frac{8}{\pi}e^{\gamma - 2} + o(1)\right) \exp\left(\frac{\pi}{2\sqrt{\mu}a}\right).
\]
We thus immediately get following.
\begin{cor}
Let $V$ be radial and assume that $V(x)(1 + |x|) \in L^1(\R^3) \cap L^{3/2}(\R^3)$, 
$\hat V \leq 0$, $\hat V(0) < 0$, 
that $\norm{V}_{L^{3/2}} < S_3$, 
and that the scattering length $a < 0$ is negative. 
Then,
\[
  \lim_{\mu \to 0} \frac{\Xi}{T_c} = \pi e^{-\gamma} \approx 1.7639,
\]
where $\gamma \approx 0.577$ is the Euler-Mascheroni constant.
\end{cor}

\noindent
That is, in the low density limit, the ratio of the energy gap and critical temperature
tends to some universal limit. This is known in the physics literature \cite{gorkov.melik-barkhudarov}.
Also, this property has been observed before in the low coupling limit \cite{hainzl.seiringer.08,bcs.original,nozieres.schmitt-rink}.
Moreover, the universal constant is the same in both the low density and weak coupling limits.

The assumptions we impose on the potential $V$ is more or less the assumptions of \cite{hainzl.seiringer.08,hainzl.seiringer.scat.length}.
The only difference is the assumption that $\norm{V}_{L^{3/2}} < S_3$ instead of the assumption that $V^{1/2}\frac{1}{p^2}|V|^{1/2}$ 
has spectrum contained in $(-1, \infty)$. As discussed above our assumption here is stronger. 
We need such a stronger assumption, since we need to control different scalings of the potential. 
As discussed in \cite{hainzl.seiringer.scat.length} our assumption captures that the operator $p^2 + V$ does not 
have any bound states.

We will follow the description of BCS theory made in 
\cite{hainzl.hamza.seiringer.solovej,hainzl.seiringer.08,hainzl.seiringer.scat.length,frank.hainzl.naboko.seiringer,braunlich.hainzl.seiringer,hainzl.seiringer.16,hainzl.seiringer.review.08}.
There the BCS gap equation at zero temperature arises as the Euler-Lagrange equations for minimisers
of the BCS functional at zero temperature
\[
	\mathcal{F}^{\mu, V}(\alpha) 
		= \frac{1}{2} \int |p^2 - \mu| \left( 1 - \sqrt{1 - 4 |\hat \alpha(p)|^2}\right)
			+ \int V(x) |\alpha(x)|^2 \ud x.
\]
For a minimiser $\alpha$ one then defines
\[
	\Delta(p) = - 2 \widehat{V \alpha}(p).
\]
This $\Delta$ then satisfies the BCS gap equation, see \cite{hainzl.seiringer.16}.
We now give the proof of our theorem.
Part of the proof is inspired by and based on the methods of \cite{hainzl.seiringer.scat.length}.

\section{Proof}
One of the key ideas in the proof is to study the asymptotics of 
\[
  m_\mu(\Delta) = \frac{1}{(2\pi)^3} \int \frac{1}{E(p)} - \frac{1}{p^2} \ud p.
\]
This is similar to what is done in \cite{hainzl.seiringer.08,hainzl.seiringer.scat.length}
for the study of the critical temperature and energy gap in the low coupling limit
and for the critical temperature in the low density limit.

The structure of the proof is as follows. 
First we find bounds on the minimiser $\alpha$ of the BCS functional. 
These then translate to bounds on the function $\Delta$, which gives some asymptotic 
behaviour of $m_\mu(\Delta)$. Armed with this, we employ the methods of \cite{hainzl.seiringer.scat.length}
to prove our theorem.

In \cite[Lemma 2]{hainzl.seiringer.08} it is proven, 
that there exists a unique minimiser $\alpha$ of the BCS functional at zero temperature 
with (strictly) positive Fourier transform. 
This we will denote by $\alpha_{\mu, V}$. 
By scaling it follows that
\[
  \mathcal{F}^{\mu, V}(\alpha) = \mu^{5/2} \mathcal{F}^{1, \sqrt{\mu}V_{\sqrt\mu}}(\beta)
  = \mu^{5/2} \left[
  	\frac{1}{2} \int |p^2 - 1| \left(1 - \sqrt{1 - 4 |\hat \beta(p)|^2}\right) \ud p
  		+ \sqrt{\mu} \int  V_{\sqrt\mu} |\beta|^2 \ud x\right],
\]
where $\beta(x) = \mu^{-3/2} \alpha(x / \sqrt{\mu})$ and $V_{\sqrt{\mu}}(x) = \mu^{-3/2} V(x/\sqrt{\mu})$. 
Note that $\norm{V_{\sqrt{\mu}}}_{L^1} = \norm{V}_{L^1}$ and that
$\norm{\sqrt{\mu} V_{\sqrt{\mu}}}_{L^{3/2}} = \norm{V}_{L^{3/2}}$.
With this, we thus see that the minimisers with positive Fourier transform satisfy
\[
  \alpha_{\mu, V}(x) = \mu^{3/2} \alpha_{1, \sqrt{\mu} V_{\sqrt{\mu}}} (\sqrt{\mu} x).
\]
We now bound this.
\begin{prop}
In the limit $\mu\to 0$ the minimiser satisfies $\norm{\alpha_{\mu, V}}_{H^1} \leq C \mu^{3/4}$.
\end{prop}
\begin{proof}
With the scaling argument above we compute
\[
  \norm{\alpha_{\mu, V}}_{H^1}^2 = \int |\hat\alpha_{\mu, V}(p)|^2 \left(1 + p^2\right) \ud p
    = \mu^{3/2} \int |\hat \alpha_{1, \sqrt{\mu}V_{\sqrt{\mu}}}(q)|^2 \left(1 + \mu q^2\right) \ud q
    \leq \mu^{3/2} \norm{ \alpha_{1, \sqrt{\mu}V_{\sqrt{\mu}}} }_{H^1}^2.
\]
We now show, that this latter norm is bounded uniformly in $\mu$.

Let $\lambda = \frac{S_3}{\norm{V}_{L^{3/2}}} > 1$. 
Then, as $\norm{\sqrt{\mu} V_{\sqrt{\mu}}}_{L^{3/2}} = \norm{V}_{L^{3/2}}$ 
it follows that $\frac{p^2}{\lambda} + \sqrt{\mu} V_{\sqrt{\mu}} \geq 0$ by Sobolev's inequality, see \cite[section. 11.3]{analysis}. 
Thus we may bound for any $\alpha$, 
\begin{align*}
 \mathcal{F}^{1, \sqrt{\mu}V_{\sqrt\mu}}(\alpha)
  & \geq \int (p^2 - 1) |\hat \alpha(p)|^2 \ud p + \int \sqrt{\mu}V_{\sqrt{\mu}}(x)|\alpha(x)|^2 \ud x
  \\
  & = \longip{\alpha}{\frac{p^2}{\lambda} + \sqrt{\mu}V_{\sqrt{\mu}}}{\alpha} + \int \left(2 \eps p^2 - 1\right) |\hat \alpha(p)|^2 \ud p
  \\
  & \geq \eps \int |\hat \alpha(p)|^2 (1 + p^2) \ud p
    +  \int (\eps p^2 - \eps - 1) |\hat \alpha(p)|^2 \ud p
  \\
  & \geq \eps \norm{\alpha}_{H^1}^2 - A,
\end{align*}
where we introduced $\eps = \frac{1}{2} - \frac{1}{2\lambda} > 0$ and 
$A = \frac{1}{4} \int \left[\eps p^2 - 1 - \eps\right]_{-} \ud p < \infty$.
Since $\mathcal{F}^{1, \sqrt\mu V_{\sqrt{\mu}}}(0) = 0$ we get for the minimiser that 
$\norm{\alpha_{1, \sqrt{\mu} V_{\sqrt{\mu}}}}_{H^1}$ is bounded
uniformly in $\mu$. We conclude that 
\[
  \norm{\alpha_{\mu, V}}_{H^1} \leq C \mu^{3/4}. 
  \qedhere
\]
\end{proof}

\begin{prop}\label{prop.small.p.vs.large.p}
For small enough $\mu$, the minimiser satisfies 
\[
	\norm{ \hat \alpha_{\mu, V} 1_{\{|p| > \eps\}} }_{L^{3/2}} \leq C \norm{ \hat \alpha_{\mu, V} 1_{\{|p| \leq \eps\}}}_{L^1}
\]
for a small $\eps > 0$ and a constant $C$, both independent of $\mu$.
\end{prop}
\begin{proof}
By the continuity of $\hat V$ we may find $\eps > 0$ such that
$ 2 \hat V(0) \leq \hat V(p) \leq \frac{1}{2} \hat V(0) < 0 $ for all $|p| \leq 2\eps$.
Let $\lambda = \frac{S_3}{\norm{V}_{L^{3/2}}} > 1$. Then $\frac{p^2}{\lambda} + V \geq 0$ and so
for the minimiser $\alpha = \alpha_{\mu, V}$ we compute
\begin{eqnarray*} 
 \mathcal{F}^{\mu, V}(\alpha)
  & = & \frac{1}{2} \int |p^2 - \mu| \left(1 - \sqrt{1 - 4\hat \alpha(p)^2}\right) \ud p
    + \int V(x) |\alpha(x)|^2 \ud x
\\
  & \geq & \int_{|p| > \eps} \left(p^2 - \mu\right) \hat \alpha(p)^2 \ud p
    + \frac{1}{(2\pi)^{3/2}} \iint \hat \alpha(p) \hat V(p - q) \hat \alpha(q) \ud p \ud q
\\
  & = & \int_{|p| > \eps} \left(p^2 - \mu\right) \hat \alpha(p)^2 \ud p
    + \frac{1}{(2\pi)^{3/2}} \left[
      \int_{|p| \leq \eps} \int_{|q|\leq \eps} \hat \alpha(p) \hat V(p - q) \hat \alpha(q) \ud p \ud q
      \right.
\\
  & & 
      \left.
      + 2 \int_{|p| \leq \eps|}\int_{|q| > \eps} \hat \alpha(p) \hat V(p - q) \hat \alpha(q) \ud p \ud q
      + \int_{|p| > \eps|}\int_{|q| > \eps} \hat \alpha(p) \hat V(p - q) \hat \alpha(q) \ud p \ud q
      \right]
\\
  & \geq & \longip{\hat \alpha 1_{\{|p| > \eps\}}}{\frac{p^2}{\lambda} + V}{\hat \alpha 1_{\{|p| > \eps\}}}
    + \int_{|p| > \eps} \left( \left(1 - \frac{1}{\lambda}\right) p^2 - \mu\right) \hat \alpha(p)^2 \ud p
\\
  & & 
      + \frac{1}{(2\pi)^{3/2}}
      \left[
        2 \hat V(0) \norm{ \hat \alpha 1_{\{|p| \leq \eps\}}}_{L^1}^2
        + 2 \int_{|p| \leq \eps} \int_{|q| > \eps} \hat \alpha(p) \hat V(p - q) \hat \alpha(q) \ud p \ud q
      \right]
\\
  & \geq & \int_{|p| > \eps} \left( \left(1 - \frac{1}{\lambda}\right) p^2 - \mu\right) \hat \alpha(p)^2 \ud p
\\
  & & 
      + \frac{1}{(2\pi)^{3/2}}
      \left[
        2 \hat V(0) \norm{ \hat \alpha 1_{\{|p| \leq \eps\}}}_{L^1}^2
        + 2 \int_{|p| \leq \eps} \int_{|q| > \eps} \hat \alpha(p) \hat V(p - q) \hat \alpha(q) \ud p \ud q
      \right].
\end{eqnarray*} 
We now bound the two remaining integrals.
For the first integral we do the following
\[
  \int_{|p| > \eps} \left( \left(1 - \frac{1}{\lambda}\right) p^2 - \mu\right) \hat \alpha(p)^2 \ud p
  \geq c \int_{|p| > \eps} \hat \alpha(p)^2 \left(1 + p^2\right) \ud p
  \geq c \norm{\hat \alpha 1_{\{|p| > \eps\}}}_{L^{3/2}}^2,
\]
by the bound $\norm{\hat g}_{L^{3/2}} \leq C\norm{g}_{H^1}$, valid for any function $g$. To see this,
simply write
\begin{align*}
\norm{\hat g}_{L^{3/2}}^{3/2}
  & = \int |\hat g(p)|^{3/2} \frac{(1 + p^2)^{3/4}}{(1 + p^2)^{3/4}} \ud p
    \leq \left(\int |\hat g(p)|^2 \left(1 + p^2 \right) \ud p \right)^{3/4} 
      \left(\int \frac{1}{(1 + p^2)^{3}} \ud p\right)^{1/4}.
\end{align*}
For the double-integral we use the Young and the Hausdorff-Young inequalities \cite[Theorems 4.2 and 5.7]{analysis}.
This gives
\begin{align*}
\abs{\int_{|p| \leq \eps} \int_{|q| > \eps} \hat \alpha(p) \hat V(p - q) \hat \alpha(q) \ud p \ud q}
  & \leq C \norm{\hat \alpha 1_{\{|p| \leq \eps\}}}_{L^1} \norm{\hat V}_{L^3} \norm{\hat \alpha 1_{\{|p| > \eps\}}}_{L^{3/2}}
  \\
  & \leq C \norm{V}_{L^{3/2}} \norm{\hat \alpha 1_{\{|p| \leq \eps\}}}_{L^1} \norm{\hat \alpha 1_{\{|p| > \eps\}}}_{L^{3/2}}.
\end{align*}
Combining all this we get the bound
\begin{align*}
\mathcal{F}^{\mu, V}(\alpha)
  & \geq 
    c \norm{\hat \alpha 1_{\{|p| > \eps\}}}_{L^{3/2}}^2
    - C_1 \norm{\hat \alpha 1_{\{|p| > \eps\}}}_{L^{3/2}} \norm{\hat \alpha 1_{\{|p| \leq \eps\}}}_{L^1}
    - C_2 \norm{\hat \alpha 1_{\{|p| \leq \eps\}}}_{L^1}^2
\end{align*}
where we absorbed the factors of $V$ into the constants $C_1, C_2 > 0$.
The right-hand-side above is a second degree polynomial in $\norm{\hat \alpha 1_{\{|p| > \eps\}}}_{L^{3/2}}$.
Moreover, the minimiser $\alpha = \alpha_{\mu, V}$ satisfies $\mathcal{F}^{\mu, V}(\alpha) \leq 0$. 
We conclude that $\norm{\hat \alpha 1_{\{|p| > \eps\}}}_{L^{3/2}}$ is between
the two roots of the second degree polynomial. In particular
\begin{align*}
\norm{\hat \alpha 1_{\{|p| > \eps\}}}_{L^{3/2}}
  & \leq 
    \frac{C_1 \norm{\hat \alpha 1_{\{|p| \leq \eps\}}}_{L^1} 
      + \sqrt{ C_1^2 \norm{\hat \alpha 1_{\{|p| \leq \eps\}}}_{L^1}^2 
      + 4 c C_2 \norm{\hat \alpha 1_{\{|p| \leq \eps\}}}_{L^1}^2 }}{
    2c
    }
\\
  & \leq 
    C \norm{\hat \alpha 1_{\{|p| \leq \eps\}}}_{L^1}. \qedhere
\end{align*}
\end{proof}

\noindent
We now bound $\Delta_{\mu, V} = - 2 \widehat{ V \alpha_{\mu, V}} = - 2 (2\pi)^{-3/2} \hat V * \hat \alpha_{\mu, V}$.
Now, $\hat V\leq 0$, and so $\Delta_{\mu, V} \geq 0$. 
By the BCS gap equation it follows that even $\Delta_{\mu, V} > 0$, see \cite{hainzl.seiringer.08}.
\begin{prop}\label{prop.bound.delta.energy.gap.low.density}
The function $\Delta_{\mu, V}$ satisfies
\[
  \Delta_{\mu, V}(p) \leq C \mu^{3/4}
  \qquad \textnormal{and}
  \qquad
  |\Delta_{\mu, V}(p') - \Delta_{\mu, V}(p)| \leq C \mu^{3/4} |p' - p|.
\] 
\end{prop}
\begin{proof}
We compute
\[
\Delta_{\mu, V}(p)
  \leq \frac{2}{(2\pi)^{3/2}} \int \abs{ \hat V(p-q) } \hat \alpha_{\mu, V}(q) \ud q
  \leq C \norm{\hat V}_{L^3} \norm{\hat \alpha_{\mu, V}}_{L^{3/2}}
  \leq C \norm{V}_{L^{3/2}} \norm{\alpha_{\mu, V}}_{H^1}
  \leq C \mu^{3/4}
\]
by the Hausdorff-Young inequality \cite[Theorem 5.7]{analysis} and 
the fact that $\norm{\hat g}_{L^{3/2}} \leq C \norm{g}_{H^1}$.
The bound for the difference is similar, using that
\begin{align*}
  \norm{\hat V(p' - \cdot) - \hat V(p - \cdot)}_{L^3}
    & \leq C \left(\int \abs{ e^{-ip'x} - e^{-ipx}}^{3/2} |V(x)|^{3/2} \ud x\right)^{2/3}
    \\
    & \leq C \left(\int \abs{p' - p}^{3/2} |x|^{3/2} |V(x)|^{3/2} \ud x\right)^{2/3}
    \\
    & = C \norm{V |\cdot|}_{L^{3/2}} |p' - p|,
\end{align*}
where we used that $\hat V(p' - \cdot) - \hat V(p - \cdot)$ is the Fourier transform of 
$\left(e^{-ip'x} - e^{-ipx}\right) V(-x)$. 
\end{proof}

\noindent
For the sake of simplifying notation, we will just write $\Delta$ for the function $\Delta_{\mu, V}$
from now on. 
With this bound on $\Delta$ we get some control over $m_\mu(\Delta)$.

By computing the spherical part of the integral, splitting the integral according to $p^2 < 2\mu$ and $p^2 > 2\mu$,
and using the substitutions $s = \frac{\mu - p^2 }{\mu}$ and $s = \frac{p^2 - \mu}{\mu}$ we may rewrite $m_\mu(\Delta)$ as 
\newlength{\templengap}
\settowidth{\templengap}{$\displaystyle 
  = \frac{\sqrt{\mu}}{4\pi^2} \,
$}
\begin{align*}
m_\mu(\Delta)
  & = 
    \frac{\sqrt{\mu}}{4\pi^2} 
      \left[
        \int_0^1 \frac{\sqrt{1 - s} - 1}{\sqrt{s^2 + \left(\frac{\Delta( \sqrt{\mu} \sqrt{1 - s})}{\mu}\right)^2}} 
          + \frac{\sqrt{1 + s} - 1}{\sqrt{s^2 + \left(\frac{\Delta( \sqrt{\mu} \sqrt{1 + s})}{\mu}\right)^2}} 
      \right.
    - \frac{1}{\sqrt{1 - s}} - \frac{1}{\sqrt{1 + s}} \ud s
  \\
  & \hspace*{\templengap}
    + \int_0^1 \frac{1}{\sqrt{s^2 + \left(\frac{\Delta( \sqrt{\mu} \sqrt{1 - s})}{\mu}\right)^2}}
        + \frac{1}{\sqrt{s^2 + \left(\frac{\Delta( \sqrt{\mu} \sqrt{1 + s})}{\mu}\right)^2}} \ud s
  \\
  & \hspace*{\templengap}
      \left.
        + \int_1^\infty \frac{\sqrt{1 + s}}{\sqrt{s^2 + \left(\frac{\Delta( \sqrt{\mu} \sqrt{1 + s})}{\mu}\right)^2}} 
          - \frac{1}{\sqrt{1 + s}} \ud s
      \right].
\end{align*}
Here by $\Delta(\sqrt{\mu}\sqrt{1 \pm s})$ we mean the value of $\Delta$ on a sphere 
with the given radius. Since $\Delta$ is radial, this is well-defined.
We now claim that
\begin{prop}\label{prop.energy.gap.low.density.expansion.m}
In the limit $\mu \to 0$ the value $m_\mu(\Delta)$ satisfies
\begin{multline*}
m_\mu(\Delta)
   = \frac{\sqrt{\mu}}{4\pi^2} 
      \left[
        \int_0^1 \frac{\sqrt{1 - s} - 1}{\sqrt{s^2 + \left(\frac{\Delta( \sqrt{\mu})}{\mu}\right)^2}} 
          + \frac{\sqrt{1 + s} - 1}{\sqrt{s^2 + \left(\frac{\Delta( \sqrt{\mu})}{\mu}\right)^2}} 
    - \frac{1}{\sqrt{1 - s}} - \frac{1}{\sqrt{1 + s}} \ud s
      \right.
  \\
      \left.
    + \int_0^1 \frac{2}{\sqrt{s^2 + \left(\frac{\Delta( \sqrt{\mu})}{\mu}\right)^2}} \ud s
    + \int_1^\infty \frac{\sqrt{1 + s}}{\sqrt{s^2 + \left(\frac{\Delta( \sqrt{\mu})}{\mu}\right)^2}} 
          - \frac{1}{\sqrt{1 + s}} \ud s
    + o(1)
      \right].
\end{multline*}
\end{prop}
\begin{proof}
For the first and last integrals this follows by a dominated convergence argument. 
One considers the difference between the claimed value and the known value
and uses a dominated convergence argument to shows that this vanishes.
For the middle integral we use propositions \ref{prop.small.p.vs.large.p} and \ref{prop.bound.delta.energy.gap.low.density}.
The argument is as follows.

Define the function(s) $x(s) = \frac{\Delta(\sqrt{1 \pm s}\sqrt{\mu})}{\mu}$. We must then show that
\[
  \lim_{\mu \to 0}\int_0^1 \frac{1}{\sqrt{s^2 + x(s)^2}} - \frac{1}{\sqrt{s^2 + x(0)^2}} \ud s = 0.
\]
First, the function $\Delta$ satisfies (with $\eps > 0$ chosen from \cref{prop.small.p.vs.large.p})
\begin{align*}
  \Delta(p) 
    & = \frac{2}{(2\pi)^{3/2}}\int \hat V(p-q) \hat \alpha_{\mu, V}(q) \ud q
    \\
    & = \frac{2}{(2\pi)^{3/2}}\int_{|q| \leq \eps} \hat V(p-q) \hat \alpha_{\mu, V}(q) \ud q
      + \frac{2}{(2\pi)^{3/2}}\int_{|q| > \eps} \hat V(p-q) \hat \alpha_{\mu, V}(q) \ud q.
\end{align*}
This gives for $|p| = \sqrt{\mu}$ that
\begin{align*}
  |\Delta(\sqrt{\mu})|
    & = \frac{2}{(2\pi)^{3/2}}\int_{|q| \leq \eps} |\hat V(p-q)| \hat \alpha_{\mu, V}(q) \ud q
      + \frac{2}{(2\pi)^{3/2}}\int_{|q| > \eps} |\hat V(p-q)| \hat \alpha_{\mu, V}(q) \ud q
    \\
    & \geq \frac{1}{(2\pi)^{3/2}} |\hat V(0)| \norm{\hat \alpha_{\mu, V} 1_{\{|p| \leq \eps\}}}_{L^1}.
\end{align*}
Also, for any $|p| = \sqrt{1 \pm s}\sqrt{\mu}$ that
\begin{align*}
  |\Delta(p)|
    & = \frac{2}{(2\pi)^{3/2}}\int_{|q| \leq \eps} |\hat V(p-q)| \hat \alpha_{\mu, V}(q) \ud q
      + \frac{2}{(2\pi)^{3/2}}\int_{|q| > \eps} |\hat V(p-q)| \hat \alpha_{\mu, V}(q) \ud q
    \\
    & \leq \frac{4}{(2\pi)^{3/2}} |\hat V(0)| \norm{\hat \alpha_{\mu, V} 1_{\{|p| \leq \eps\}}}_{L^1}
      + \frac{2}{(2\pi)^{3/2}} \norm{\hat V}_{L^3} \norm{ \hat \alpha_{\mu, V} 1_{\{|p| > \eps\}}}_{L^{3/2}}
    \\
    & \leq C \norm{\hat \alpha_{\mu, V} 1_{\{|p| \leq \eps\}}}_{L^1} \leq C |\Delta(\sqrt{\mu})|,
\end{align*}
by the Hausdorff-Young inequality \cite[Theorem 5.7]{analysis} and \cref{prop.small.p.vs.large.p} above. 
Thus, the function(s) $x(s)$ satisfies the desired $|x(s)| \leq C |x(0)|$.
With this, we may now prove the desired convergence of integrals.
\begin{align*}
 \abs{\frac{1}{\sqrt{s^2 + x(s)^2}} - \frac{1}{\sqrt{s^2 + x(0)^2}}}
    & = \frac{\abs{x(s)^2 - x(0)^2}}{\sqrt{s^2 + x(s)^2}\sqrt{s^2 + x(0)^2}\left(\sqrt{s^2 + x(s)^2} + \sqrt{s^2 + x(0)^2}\right)}
    \\
    & \leq \frac{ C \mu^{1/4} s |x(0)|}{\sqrt{s^2 + x(s)^2}\sqrt{s^2 + x(0)^2} \left( s + \sqrt{s^2 + x(0)^2}\right)}
    \\
    & \leq C\mu^{1/4} \frac{|x(0)|}{\sqrt{s^2 + x(0)^2} \left(s + \sqrt{s^2 + x(0)^2}\right)}.
\end{align*}
Now, one may compute that
\[
  \int_0^1 \frac{|x(0)|}{\sqrt{s^2 + x(0)^2}\left(s + \sqrt{s^2 + x(0)^2}\right)} \ud s = O(1).
\]
This shows that
\[
  \int_0^1 \frac{1}{\sqrt{s^2 + x(s)^2}} - \frac{1}{\sqrt{s^2 + x(0)^2}} \ud s
    = O \left(\mu^{1/4}\right)
\]
vanishes as desired. We conclude the desired.
\end{proof}

\noindent
The remainder of this paper uses the methods of \cite{hainzl.seiringer.scat.length}. 
We decompose 
\[
  B_\Delta 
    := V^{1/2} \frac{1}{E} |V|^{1/2}
    = V^{1/2} \frac{1}{p^2} |V|^{1/2} + m_\mu(\Delta) \ket{V^{1/2}} \bra{|V|^{1/2}} + A_{\Delta, \mu},
\]  
where $A_{\Delta, \mu}$ is defined such that this holds. 
That is, its kernel is
\[
  A_{\Delta, \mu}(x,y) 
    = V(x)^{1/2} |V(y)|^{1/2} \frac{1}{2\pi^2} \int_0^\infty \left( \frac{\sin p |x-y|}{p |x-y|} - 1\right)
      \left(\frac{1}{E} - \frac{1}{p^2}\right) p^2 \ud p.
\] 
In order to see this, note that $\int_{S^2} e^{ipx} \ud p = 4\pi \frac{\sin |x|}{|x|}$.
The operator $B_\Delta$ is the Birman-Schwinger operator associated to $E + V$. 
One easily checks that $E + V$ has its lowest eigenvalue $0$, see \cite{hainzl.seiringer.08}.
(This follows from the fact that $\hat V\leq 0$ is negative and so the ground state of $E + V$ 
can be chosen to have non-negative Fourier transform. Hence it is not orthogonal to $\alpha_{\mu, V}$,
which is an eigenfunction with eigenvalue $0$.)
Thus $B_\Delta$ has $-1$ as its lowest eigenvalue.
\begin{prop}
In the limit $\mu \to 0$ the function $\Delta$ satisfies $\Delta(\sqrt{\mu}) = o(\mu)$.
\end{prop}
\begin{proof}
Suppose for contradiction that $\frac{\Delta(\sqrt{\mu})}{\mu}$ does not vanish. That is, 
suppose that there is some subsequence with $\Delta(\sqrt{\mu}) > B\mu$ for $\mu \to 0$ for some constant $B > 0$. 
We use the decomposition
\[
  B_\Delta 
    = V^{1/2} \frac{1}{p^2} |V|^{1/2} + m_\mu(\Delta) \ket{V^{1/2}} \bra{|V|^{1/2}} + A_{\Delta, \mu}.
\]  
By the assumptions on $V$, the spectrum of $V^{1/2}\frac{1}{p^2}|V|^{1/2}$ is contained
in $(-1, \infty)$. We show that the remaining two terms in the decomposition above 
vanish in the limit $\mu \to 0$, and so that the spectrum 
of $B_\Delta$ approaches that of $V^{1/2}\frac{1}{p^2}|V|^{1/2}$. 
Since the latter has its lowest eigenvalue strictly larger
than $-1$, we get a contradiction. 

For $m_\mu(\Delta)$ we use \cref{prop.energy.gap.low.density.expansion.m} above. 
The only term that does not immediately vanish
in the limit $\mu \to 0$ is the term
\[
  \frac{\mu^{1/2}}{4\pi^2} 
  	\int_1^\infty \frac{\sqrt{1 + s}}{\sqrt{s^2 + \left(\frac{\Delta(\sqrt{\mu})}{\mu}\right)^2}} - \frac{1}{\sqrt{1 + s}} \ud s.
\]
By splitting this integral according to $s < \frac{\Delta(\sqrt{\mu})}{\mu}$ 
and $s > \frac{\Delta(\sqrt{\mu})}{\mu}$ we see that this term may be bounded by
$C\mu^{-1/2} \Delta(\sqrt{\mu}) \leq C\mu^{1/4}$ by \cref{prop.bound.delta.energy.gap.low.density}. 
Hence this term indeed also vanishes. 

For the kernel of $A_{\Delta, \mu}$ we use that 
$\abs{\frac{\sin b}{b} - 1} \leq C \min\{1, b^2\} \leq C b^{\gamma}$ 
for any $0\leq \gamma \leq 2$ for the specific choice of $\gamma = \frac{1}{2}$. 
Then
\begin{multline*}
  \abs{A_{\Delta, \mu}(x,y)}
    \leq  
    C |V(x)|^{1/2} |V(y)|^{1/2} |x-y|^{1/2}\left[\int_0^{\sqrt{2\mu}} \abs{\frac{1}{E} - \frac{1}{p^2}} p^{5/2} \ud p
      + \int_{\sqrt{2\mu}}^\infty \abs{\frac{1}{E} - \frac{1}{p^2}} p^{5/2} \ud p \right].
\end{multline*}
For the first integral we bound $E(p) \geq \Delta(p) \geq B\mu$ and so
\[
  \int_0^{\sqrt{2\mu}} \abs{\frac{1}{E} - \frac{1}{p^2}} p^{5/2} \ud p
    \leq \int_0^{\sqrt{2\mu}} \frac{1}{B\mu} (2\mu)^{5/4} + (2\mu)^{1/4} \ud p
    \leq C \mu^{3/4}.
\]
We bound the second integral as follows. First, with the substitution $s = \frac{p^2 - \mu}{\mu}$ we get
\begin{align*}
\int_{\sqrt{2\mu}}^\infty \abs{\frac{1}{E_\Delta} - \frac{1}{p^2}} p^{5/2} \ud p
  & = \frac{\mu^{3/4}}{2} \int_1^\infty 
      \abs{ \frac{1 + s}{\sqrt{s^2 + \left(\frac{\Delta(\sqrt{\mu}\sqrt{1 + s})}{\mu}\right)^2}} - 1}
        \frac{1}{(1 + s)^{1/4}} \ud s
  \\
  & \leq \frac{\mu^{3/4}}{2} \int_1^\infty \frac{1}{s(1 + s)^{1/4}} 
        + \frac{\sqrt{s^2 + \left(\frac{\Delta(\sqrt{\mu}\sqrt{1 + s})}{\mu}\right)^2} - s}{s(1+s)^{1/4}} \ud s
  \\
  & \leq C \mu^{3/4} + C\mu^{-1/4} \int_1^\infty \frac{\Delta(\sqrt{\mu}\sqrt{1 + s})}{s(1 + s)^{1/4}} \ud s
  \\
  & \leq C \mu^{1/2},
\end{align*}
where we used that $|\Delta(p)| \leq C\mu^{3/4}$.
The integral 
$\iint |V(x)||V(y)||x-y| \ud x \ud y < \infty$ is finite 
by the assumptions on $V$. 
Thus $\norm{A_{\Delta, \mu}}_{2}\leq C\mu^{1/2}$ vanishes as desired. 
\end{proof}

\noindent
Using this refined bound, $\Delta(\sqrt{\mu}) = o(\mu)$, we may use a dominated convergence
argument to show that
\begin{multline*}
m_\mu(\Delta)
  = \frac{\sqrt{\mu}}{4\pi^2} 
      \left[
        \int_0^1 \frac{\sqrt{1 - s} - 1}{s} 
          + \frac{\sqrt{1 + s} - 1}{s} 
    - \frac{1}{\sqrt{1 - s}} - \frac{1}{\sqrt{1 + s}} \ud s
      \right.
  \\
      \left.
    + \int_0^1 \frac{2}{\sqrt{s^2 + \left(\frac{\Delta( \sqrt{\mu})}{\mu}\right)^2}} \ud s
    + \int_1^\infty \frac{\sqrt{1 + s}}{s} 
          - \frac{1}{\sqrt{1 + s}} \ud s
    + o(1)
      \right].
\end{multline*}
These integrals can be computed (somewhat easily by hand). 
This is done in \cite{hainzl.seiringer.08}.
We conclude that 
\begin{equation}\label{eqn.m.asymptotics}
  m_\mu(\Delta) = \frac{\sqrt{\mu}}{2\pi^2}\left( \log \frac{\mu}{\Delta(\sqrt{\mu})} - 2 + \log 8 + o(1)\right)
\end{equation}
in the limit $\mu \to 0$. 
In particular $m_\mu(\Delta) \gg \sqrt{\mu}$.
Now, we are interested in bounding $A_{\Delta, \mu}$.
\begin{prop}\label{prop.energy.gap.low.density.bound.A}
The operator $A_{\Delta, \mu}$ vanishes in the following sense. 
\[
  \lim_{\mu \to 0} \frac{\norm{A_{\Delta, \mu}}_2}{m_\mu(\Delta)} = 0.
\]
\end{prop}
\begin{proof}
The proof is similar as above, only we give a more refined bound on the kernel. 
We bound the $\frac{\sin b}{b}$ term by
\begin{multline*}
  \abs{\frac{\sin |p||x-y|}{|p||x-y|} - 1}
    \\
    \leq C \left[p^2 Z^2 1_{\{|x-y| < Z\}} + |p|^{1/2}|x-y|^{1/2} 1_{\{|x-y| > Z\}} \right]
    1_{\{p^2 < 2\mu\}} 
    + C |p|^{1/2}|x-y|^{1/2} 1_{\{p^2 > 2\mu\}}.
\end{multline*}
Where $Z > 0$ is arbitrary, and the constant $C$ does not depend on $Z$.
Then
\begin{multline*}
\abs{A_{\Delta, \mu}(x,y)}
  \leq C |V(x)|^{1/2}|V(y)|^{1/2} 
    \left[ 
      Z^2 \int_0^{\sqrt{2\mu}} \abs{\frac{1}{E_\Delta} - \frac{1}{p^2}} p^4 \ud p
    \right.
  \\
    \left.
      + |x-y|^{1/2} 1_{\{|x-y| > Z\}} \int_0^{\sqrt{2\mu}} \abs{\frac{1}{E_\Delta} - \frac{1}{p^2}} p^{5/2} \ud p
      + |x-y|^{1/2} \int_{\sqrt{2\mu}}^{\infty} \abs{\frac{1}{E_\Delta} - \frac{1}{p^2}} p^{5/2} \ud p
    \right].
\end{multline*}
Now, the first and second integral may be bounded by 
$ m_\mu(\Delta) \mu$ and $ m_\mu(\Delta) \mu^{1/4}$ 
as follows. For any $\gamma$ we may bound
\[
  \int_0^{\sqrt{2\mu}} \abs{\frac{1}{E_\Delta} - \frac{1}{p^2}} p^\gamma \ud p
    \leq \int_0^{\sqrt{2\mu}} \left(\frac{1}{E_\Delta} - \frac{1}{p^2}\right) p^\gamma + 2 p^{\gamma - 2} \ud p
    \leq C m_\mu(\Delta) \mu^\frac{\gamma - 2}{2}.
\]
Similarly as before, the last integral may be bounded by $\mu^{1/2} \ll  m_\mu(\Delta)$. 
Again, by the assumptions on $V$ it follows that $\iint |V(x)||V(y)||x-y| \ud x \ud y < \infty$ is finite. 
Thus we get
\[
  \lim_{\mu\to 0} \frac{\norm{A_{\Delta, \mu}}_2}{ m_\mu(\Delta)} = 0.
  \qedhere
\]
\end{proof} 

\noindent
We may decompose 
\[
  1 + B_\Delta = \left(1 + V^{1/2} \frac{1}{p^2}|V|^{1/2}\right) 
    \left(1 + \frac{m_\mu(\Delta)}{1 + V^{1/2} \frac{1}{p^2}|V|^{1/2}} 
      \left( \ket{V^{1/2}}\bra{|V|^{1/2}} + \frac{A_{\Delta,\mu}}{m_\mu(\Delta)}\right)\right).
\]
Since $-1$ is an eigenvalue of $B_\Delta$ we get that $-1$ is an eigenvalue of 
\[
  \frac{m_\mu(\Delta)}{1 + V^{1/2} \frac{1}{p^2}|V|^{1/2}} 
    \left( \ket{V^{1/2}}\bra{|V|^{1/2}} + \frac{A_{\Delta,\mu}}{m_\mu(\Delta)}\right).
\] 
\Cref{prop.energy.gap.low.density.bound.A} 
above gives that the term $\frac{A_{\Delta,\mu}}{m_\mu(\Delta)}$ vanishes in the limit $\mu \to 0$.
The other term has rank one and thus we get
\[
  \lim_{\mu\to 0} \frac{-1}{m_\mu(\Delta)} 
    = \longip{|V|^{1/2}}{\frac{1}{1 + V^{1/2} \frac{1}{p^2}|V|^{1/2}}}{V^{1/2}} = 4\pi a.
\]
We now show that the rate of convergence is $o(\mu^{1/2})$.

First, we improve on \cref{prop.energy.gap.low.density.bound.A}. 
Since $m_\mu(\Delta)$ is of order 1 in the limit $\mu \to 0$
we get for the third integral in the proof of \cref{prop.energy.gap.low.density.bound.A} that
\[
  \int_{\sqrt{2\mu}}^\infty \abs{\frac{1}{E_\Delta} - \frac{1}{p^2}} p^{5/2} \ud p
    \leq C \mu^{1/2} \ll \mu^{1/4}.
\]
Hence, for any $Z > 0$ and a constant $C$, that does not depend on $Z$ we get the bound
\[
  \limsup_{\mu\to 0} \frac{\norm{A_{\Delta, \mu}}_2}{\mu^{1/4}}
    \leq C \left(\iint_{\{|x-y| > Z\}} |V(x)||V(y)||x-y| \ud x \ud y\right)^{1/2}.
\]
By the assumptions on $V$, the integrand here is integrable and so
taking $Z \to \infty$ we get that
\begin{equation}\label{eqn.A.bound.improved}
  \lim_{\mu\to 0} \frac{\norm{A_{\Delta, \mu}}_2}{\mu^{1/4}} = 0.
\end{equation}
Additionally, $A_{\Delta, \mu}$ vanishes in the limit $\mu \to 0$.
Thus the operator
\[
  1 + V^{1/2}\frac{1}{p^2}|V|^{1/2} + A_{\Delta, \mu}
\] 
is invertible for small $\mu$ and so we may write
\[
  1 + B_\Delta = \left(1 + V^{1/2}\frac{1}{p^2}|V|^{1/2} + A_{\Delta, \mu}\right)
    \left(1 + \frac{ {m}_\mu(\Delta)}{1 + V^{1/2}\frac{1}{p^2}|V|^{1/2} + A_{\Delta, \mu}} \ket{V^{1/2}}\bra{|V|^{1/2}}\right).
\]
Since $-1$ is an eigenvalue of $B_\Delta$ we get that $-1$ is an eigenvalue of the latter operator.
This has rank one and so we get that
\begin{equation}\label{eqn.energy.gap.low.density.scatt}
  \frac{-1}{m_\mu(\Delta)} = \longip{|V|^{1/2}}{\frac{1}{1 + V^{1/2}\frac{1}{p^2}|V|^{1/2} + A_{\Delta, \mu}}}{V^{1/2}}.
\end{equation}
We decompose the middle operator on the right-hand-side as 
\begin{multline*}
\frac{1}{1 + V^{1/2} \frac{1}{p^2}|V|^{1/2} + A_{\Delta, \mu}}
\\	
	=
    \frac{1}{1 + V^{1/2} \frac{1}{p^2}|V|^{1/2}} 
    - 
      \frac{1}{1 + V^{1/2} \frac{1}{p^2}|V|^{1/2}} 
        A_{\Delta, \mu} 
      \frac{1}{1 + V^{1/2} \frac{1}{p^2}|V|^{1/2}}
  \\
    + \frac{1}{1 + V^{1/2} \frac{1}{p^2}|V|^{1/2}} 
        A_{\Delta, \mu} 
      \frac{1}{1 + V^{1/2} \frac{1}{p^2}|V|^{1/2} + A_{\Delta, \mu}} 
        A_{\Delta, \mu}
      \frac{1}{1 + V^{1/2} \frac{1}{p^2}|V|^{1/2}},
\end{multline*}
which is perhaps most easily seen by writing the left-hand-side as a power series in $A_{\Delta, \mu}$.
Plugging this into \cref{eqn.energy.gap.low.density.scatt} above 
we get $4 \pi a$ for the first term.
The second term gives 
\[
  \longip{f}{\sgn V A_{\Delta, \mu}}{f}, \qquad \textnormal{with} \qquad 
    f = \frac{1}{1 + V^{1/2}\frac{1}{p^2}|V|^{1/2}} V^{1/2}.
\]
This function $f$ is the same function $f$, which was studied in \cite{hainzl.seiringer.scat.length}.
There it was shown that $f$ satisfies $f(x) |V(x)|^{1/2} (1 + |x|) \in L^1$.

The third term in the expansion above is $o(\mu^{1/2})$ by \cref{eqn.A.bound.improved} above.
We show that the second term is $o(\mu^{1/2})$ as well.
\begin{prop}
In the limit $\mu \to 0$ we have
$\longip{f}{\sgn V A_{\Delta, \mu}}{f} = o(\mu^{1/2})$.
\end{prop}
\begin{proof}
This is similar to the bound on $A_{\Delta, \mu}$ above. 
We bound the kernel of $A_{\Delta, \mu}$ by
\begin{multline*}
|A_{\Delta, \mu}(x,y)|
  \leq C|V(x)|^{1/2} |V(y)|^{1/2} 
      \left[
    Z^2 \int_0^{\sqrt{2\mu}} \abs{\frac{1}{E_\Delta} - \frac{1}{p^2}} p^4 \ud p
      \right.
  \\
      \left.
    + |x-y|1_{\{|x-y| > Z\}} \int_0^{\sqrt{2\mu}} \abs{\frac{1}{E_\Delta} - \frac{1}{p^2}} p^3 \ud p
    + |x-y|^{3/4} \int_{\sqrt{2\mu}}^\infty \abs{\frac{1}{E_\Delta} - \frac{1}{p^2}}p^{11/4} \ud p
      \right].
\end{multline*}
These integrals are bounded by $\mu, \mu^{1/2}$ and $\mu^{5/8}$ respectively similarly as in 
\cref{prop.energy.gap.low.density.bound.A} above. (Recall that $m_\mu(\Delta)$ is of order 1.)
Thus
\[
  \limsup_{\mu \to 0} \frac{\abs{\longip{f}{\sgn V A_{\Delta, \mu}}{f}}}{\mu^{1/2}} 
    \leq C \iint_{\{|x-y| > Z\}} |f(x)||V(x)|^{1/2} |x-y| |f(y)| |V(y)|^{1/2} \ud x \ud y.
\]
Since $f(x) |V(x)|^{1/2}(1 + |x|) \in L^1$ we get the desired by taking $Z \to \infty$.
\end{proof}

\noindent
We thus conclude that 
\[
  m_\mu(\Delta) = \frac{-1}{4\pi a} + o \left(\mu^{1/2}\right).
\]
With the asymptotics of $m_\mu(\Delta)$ above, \cref{eqn.m.asymptotics}, we thus get
\[
  \lim_{\mu \to 0} \left( \log \frac{\mu}{\Delta(\sqrt{\mu})} + \frac{\pi}{2 \sqrt{\mu} a}\right) = 2 - \log 8.
\]
Now, we want to replace $\Delta(\sqrt{\mu})$ by the energy gap $\Xi = \inf E$. 
Clearly $\Xi \leq \Delta(\sqrt{\mu})$. On the other hand
\[
  \Xi \geq \min_{|p^2 - \mu| \leq \Xi} \Delta(p) \geq \Delta(\sqrt{\mu}) ( 1 + o(1))
\]
Thus we conclude the desired
\[
  \lim_{\mu \to 0} \left( \log \frac{\mu}{\Xi} + \frac{\pi}{2 \sqrt{\mu} a}\right) = 2 - \log 8.
\]
This concludes the proof of \cref{thm.energy.gap.low.density}.

\subsection*{Acknowledgements}
Most of this work was done as part of the author's master's thesis. 
The author would like to thank Jan Philip Solovej for his supervision of this process.

\bibliographystyle{spmpsci}
\bibliography{bibliography}

\begin{thebibliography}{10}
\providecommand{\url}[1]{{#1}}
\providecommand{\urlprefix}{URL }
\expandafter\ifx\csname urlstyle\endcsname\relax
  \providecommand{\doi}[1]{DOI~\discretionary{}{}{}#1}\else
  \providecommand{\doi}{DOI~\discretionary{}{}{}\begingroup
  \urlstyle{rm}\Url}\fi

\bibitem{bcs.original}
Bardeen, J., Cooper, L.N., Schrieffer, J.R.: Theory of superconductivity.
\newblock Phys. rev. \textbf{108}(5) (1957)

\bibitem{braunlich.hainzl.seiringer}
Bräunlich, G., Hainzl, C., Seiringer, R.: Translation-invariant quasi-free
  states for fermionic systems and the bcs approximation.
\newblock Reviews in Mathematical Physics \textbf{26}(07), 1450012 (2014).
\newblock \doi{10.1142/S0129055X14500123}.
\newblock \urlprefix\url{https://doi.org/10.1142/S0129055X14500123}

\bibitem{frank.hainzl.naboko.seiringer}
Frank, R., Hainzl, C., Naboko, S., Seiringer, R.: The critical temperature for
  the bcs equation at weak coupling.
\newblock J. Geom. Anal. \textbf{17} (2007).
\newblock \doi{10.1007/BF02937429}

\bibitem{gorkov.melik-barkhudarov}
Gor'kov, L.P., Melik-Barkhudarov, T.K.: Contributions to the theory of
  superfluidity in an imperfect fermi gas.
\newblock Soviet Physics JETP \textbf{13}(5) (1961)

\bibitem{hainzl.hamza.seiringer.solovej}
Hainzl, C., Hamza, E., Seiringer, R., Solovej, J.P.: The bcs functional for
  general pair interactions.
\newblock Commun. Math. Phys \textbf{281} (2008).
\newblock \doi{10.1007/s00220-008-0489-2}

\bibitem{hainzl.seiringer.scat.length}
Hainzl, C., Seiringer, R.: The bcs critical temperature for potentials with
  negative scattering length.
\newblock Lett. Math. Phys \textbf{84}, 99--107 (2008)

\bibitem{hainzl.seiringer.08}
Hainzl, C., Seiringer, R.: Critical temperature and energy gap for the bcs
  equation.
\newblock Physical Review B \textbf{77} (2008).
\newblock \doi{10.1103/PhysRevB.77.184517}

\bibitem{hainzl.seiringer.review.08}
Hainzl, C., Seiringer, R.: Spectral properties of the bcs gap equation of
  superfluidity  (2008).
\newblock \doi{10.1142/9789812832382_0009}

\bibitem{hainzl.seiringer.16}
Hainzl, C., Seiringer, R.: The bardeen–cooper–schrieffer functional of
  superconductivity and its mathematical properties.
\newblock Journal of Mathematical Physics \textbf{57}(2), 021101 (2016).
\newblock \doi{10.1063/1.4941723}.
\newblock \urlprefix\url{https://doi.org/10.1063/1.4941723}

\bibitem{leggett.diatomic}
Leggett, A.J.: Diatomic molecules and cooper pairs.
\newblock In: A.~Pekalski, J.A. Przystawa (eds.) Modern Trends in the Theory of
  Condensed Matter, pp. 13--27. Springer Berlin Heidelberg, Berlin, Heidelberg
  (1980)

\bibitem{analysis}
Lieb, E.H., Loss, M.: Analysis, 2. ed. edn.
\newblock Graduate studies in mathematics ; 14. American Mathematical Society,
  Providence, R.I (2001)

\bibitem{nozieres.schmitt-rink}
Nozières, P., Schmitt-Rink, S.: Bose condensation in an attractive fermion
  gas: From weak to strong coupling superconductivity.
\newblock Journal of Low Temperature Physics \textbf{59}, 195–211 (1985)

\end{thebibliography}
\end{document}